\documentclass[12pt]{article}

\usepackage[english]{babel}
\usepackage[utf8x]{inputenc}
\usepackage[T1]{fontenc}

\usepackage[a4paper,top=3cm,bottom=2cm,left=3cm,right=3cm,marginparwidth=1.75cm]{geometry}

\usepackage{amsmath,amssymb}
\usepackage{graphicx}
\usepackage[colorinlistoftodos]{todonotes}
\usepackage[colorlinks=true, allcolors=black]{hyperref}
\usepackage{amsthm, mathtools,amssymb}
\usepackage{graphicx,color}
\usepackage{tabularx}
\usepackage{tikz}
\usepackage{pgfplots}
\usepackage{pst-plot}
\usepackage{pstricks}
\usepgfplotslibrary{fillbetween}
\usetikzlibrary{patterns}
\usetikzlibrary{shadings}
\usetikzlibrary{arrows}
\usetikzlibrary{calc}
\usetikzlibrary{decorations.markings}
\DeclareGraphicsExtensions{.pdf,.eps,.fig,.pdf_tex}
\usepackage{caption,subcaption,float}
\theoremstyle{plain} \numberwithin{equation}{section}
\newtheorem{theorem}{Theorem}[section]

\theoremstyle{definition}
\newtheorem{definition}[theorem]{Definition}

\newtheorem{example}[theorem]{Example}
\newcommand\sgn{\operatorname{sgn}}
\renewcommand\H{\mathcal{H}}

\title{Hyperwalk Formulae for Even and Odd Laplacians in Finite CW-Hypergraphs}
\author{Ivan Contreras$^{1}$ \and Sarah Loeb$^{2}$ \and Chengzheng Yu$^3$}

\begin{document}
\maketitle

\footnotetext[1]{Department of Mathematics, University of Illinois at Urbana-Champaign, Urbana, IL 61801, \texttt{icontrer@illinois.edu}}
\footnotetext[2]{Department of Mathematics, College of William and Mary, Williamsburg, VA 23187, \texttt{sjloeb@wm.edu}}
\footnotetext[3]{Graduate School of Arts and Sciences, Georgetown University, Washington D.C., 20057,  \texttt{cy375@georgetown.edu}}

\begin{abstract}
In this note we provide a combinatorial interpretation for the powers of the hypergraph Laplacians. Our motivation comes from the discrete formulation of quantum mechanics and thermodynamics in the case of finite graphs, which suggest a natural extension to simplicial and CW-complexes. With this motivation, we also define generalizations of the odd Laplacian which is specific to hypergraphs arising from CW-complexes. We then provide a combinatorial interpretation for the powers of these Laplacians.
\end{abstract}

\section{Introduction}

In graph theory, the adjacency matrix plays a basic and fundamental role in  describing the combinatorics of a graph. It is a simple observation that the $k$-th power of the adjacency matrix identifies the number of walks of length $k$ between every two vertices. We may view this fact as the adjacency matrix giving a generating function for the number of walks in a graph. Using a similar argument, Yu~\cite{Chengzheng} defines generalized walks and shows these combinatorial objects are counted by the powers of the Laplacian matrix. In addition, he gives a combinatorial interpretation for the odd Laplacian, using the terminology from supersymmetric quantum mechanics, adapted to graphs (see e.g.~\cite{Mnev2,delVeccio}).

This approach is inspired by the partition function of quantum mechanics for a free particle confined in a finite graph $\Gamma$. The space of states $\mathcal H_{V(\Gamma)}$ of the system is the finite dimensional vector space $\mathbb C^{\vert V(\Gamma) \vert} $. Its supersymmetric version includes states on edges, so 
$\mathcal H_{\Gamma}=\mathbb C^{\vert V(\Gamma) \vert} \oplus \mathbb C^ {\vert E(\Gamma)\vert }$. The even and odd Laplacian matrices act in the first and second components respectively, and the partition function for supersymmetric quantum mechanics on a graph is thus defined as $Z_{\Gamma}(t)= \exp(\frac{-i}{\hbar} \Delta_{\Gamma})$, where $\Delta_{\Gamma}=\Delta^{+}\oplus \Delta^{-}$, i.e.\ the direct sum of the even and odd Laplacian operators.

This approach can be naturally extended to more general combinatorial objects, such as hypergraphs. A \emph{finite hypergraph} $\mathcal{H}$ is a finite set $V(\H)$ of \emph{vertices} and a set $E(\H)$ of \emph{edge}, which are non-empty subsets of $V(\H)$. Several spectral graph theoretical results have been extended to hypergraphs (see e.g.~\cite{Spec}).

In Section~\ref{sec_HG}, we extend the combinatorial interpretation of the even and odd Laplacian matrices as generating functions of walks on hypergraphs. Topological spaces which admit a CW-decomposition are natural examples of hypergraphs, for which the edges are precisely the cells of the decomposition. Motivated by extending the partition function for graph quantum mechanics to CW-complexes, in Section~\ref{sec_CW} we consider hypergraphs which arrise from CW-complexes. The advantage of dealing with these CW-hypergraphs is that we have a natural notion of orientation on edges, which comes from the orientation of cells in the complex. We use this notion to define Laplacians in terms of the oriented incidence, which are the boundary operators of the CW-complex. This allows us to define a signed edge-version of hyperwalks of a hypergraph $\mathcal H$ and compute the generating function for the number signed walks (Theorem~\ref{d+1_d_thm}). 

\section{Hyperwalks and edge-hyperwalks} \label{sec_HG}

Let $\mathcal{H}$ be a hypergraph. The \emph{incidence matrix} $I$ of $\mathcal{H}$ is an $n \times m$-matrix where \[I(i,j)=\begin{cases}
1 & \text{if } v_{i} \in e_{j},\\  
0 & \text{otherwise.}
\end{cases}\]

Let $\mathcal{H}$ be a hypergraph, and let $I$ be the incidence matrix of $\mathcal{H}$. The \emph{even hypergraph Laplacian} $\Delta^+(\H)$ of $\mathcal{H}$ is $II^{t}$. The \emph{odd hypergraph Laplacian} $\Delta^-(\H)$ of $\mathcal{H}$ is $I^{t}I$.
Note that the definition of incidence matrix (and subsequently the definition of Laplacian matrices) depends on the labeling of the edges and vertices. However, the resulting counting formulae will be consistent with the choice of the labeling.

\begin{definition}\label{def_hyperwalk}

A \emph{hyperwalk} in a hypergraph $\H$ is a sequence $v_{0},e_{1},v_{1},e_{2},\ldots,v_{k}$ of vertices $v_{i}$ and edges $e_{j}$, such that $v_{i-1} \in e_{i}$ and $v_{i} \in e_{i}$ for $1 \leq i \leq k$. The \emph{length} of a hyperwalk is its number of edges.

\end{definition}

Note that the definition of a hyperwalk does not require consecutive vertices to be distinct. This definition extends in a natural way the notion of generalized walks as defined in~\cite{Chengzheng} for finite graphs. 

\begin{definition}\label{def_edge_hyperwalk}
An \emph{edge-hyperwalk} in a hypergraph $\H$ is a sequence $e_{0},v_{1},e_{1},v_{2},\ldots,e_{k}$ of edges $e_{i}$ and vertices $v_{j}$, such that $v_{i} \in e_{i-1} \cap e_{i}$ $1 \leq i \leq k$. The \emph{length} of an edge-hyperwalk is its number of vertices.
\end{definition}

As in the vertex case, we do not require consecutive edges to be distinct. Example~\ref{example} illustrates of both types of hyperwalks. 

Theorem~\ref{hyper_walk_thm} gives a combinatorial interpretation of powers of the even Laplacian. It also gives the generating function for the number of hyperwalks on a CW-hypergraph $\mathcal H$. The theorem is a natural generalization of the well-known formula for the number of walks on a finite graph in terms of the powers of the adjacency matrix. A similar formula for regular graphs can be found in~\cite{Mnev2} and for general graphs in~\cite{delVeccio}.

\begin{theorem}\label{hyper_walk_thm}

$(\Delta^{+})^{k}(i,j)$ is the number of hyperwalks in $\mathcal{H}$ from $v_{i}$ to $v_{j}$ of length $k$.

\end{theorem}

\begin{proof}

We proceed by induction on the exponent $k$. Let $V(\H) = \{v_1,\ldots,v_n\}$ and $E(\H) = \{e_1,\ldots,e_m\}$. For $k=1$, we have, $(\Delta^+)^1 = \Delta^+ = II^t$, where $I$ is the incidence matrix of $\H$.  

Thus, $\Delta^{+}(i,j)=\sum_{q=1}^{m}I_{iq}I_{jq}$. Since $I$ is a $(0,1)$-matrix, all contributions to the sum are also either 0 or 1. In particular, if $e_q$ contains both $v_i$ and $v_j$, then the term $I_{iq} I_{jq}$ is $1$ and otherwise it is $0$. This corresponds exactly to when there is a hyperwalk of length $1$ from $v_i$ to $v_j$ using $e_q$. The summation over the edges thus correctly counts the number of hyperwalks of length 1 from $v_i$ to $v_j$. 

Suppose the claim holds for $k \ge 1$.

We can view $(\Delta^+)^{k+1}$ as $(\Delta^+)^k \Delta^+$. Thus, 
\[ (\Delta^+)^{k+1}(i,j) = \sum_{l=1}^m (\Delta^+)^k (i,l) \cdot \Delta^+(l,j). \]
Any hyperwalk of length $k+1$ may be viewed as a hyperwalk of length $k$ from $v_i$ to some vertex $v_q$ followed by a hyperwalk of length 1 from $v_q$ to $v_j$. Since the number of such hyperwalks are counted by $(\Delta^+)^{k}(i,q)$ and $\Delta^+ (q,j)$ respectively, the product of the terms gives the number of hyperwalks of length $k+1$ going from $v_i$ to $v_j$ that have penultimate vertex $v_q$. Summing over all vertices thus gives the total number of hyperwalks of length $k+1$ from $v_i$ to $v_j$. 

\end{proof}

In a similar way, Theorem~\ref{edge_hyper_walk_thm} gives a combinatorial interpretation of the odd Laplacian matrix and determines the generating function for the number of edge-hyperwalks. From the physical point of view, the powers of the odd Laplacian appear in the Feynman expansion of the partition function of quantum mechanics for CW-complexes~\cite{Mnev2,delVeccio}. A similar expression for the case of graphs was proven in~\cite{Chengzheng}.

\begin{theorem}\label{edge_hyper_walk_thm}

$(\Delta^{-})^{k}(i,j)$ is the number of the edge-hyperwalks of $\mathcal{H}$ from $e_{i}$ to $e_{j}$ of length $k$.

\end{theorem}

\begin{proof}

Similar to the proof of theorem~\ref{hyper_walk_thm}, we proceed by induction on the exponent $k$. Let $V(\H) = \{v_1,\ldots,v_n\}$ and $E(\H) = \{e_1,\ldots,e_m\}$. For $k=1$, we have, $(\Delta^-)^1 = \Delta^- = I^tI$, where $I$ is the incidence matrix of $\H$.

Thus, $\Delta^{-}(i,j)=\sum_{q=1}^{m}I_{iq}I_{jq}$. Since $I$ is a $(0,1)$-matrix, all contributions to the sum are also either 0 or 1. In particular, if $v_q$ is contained in both $e_i$ and $e_j$, then the term $I_{iq} I_{jq}$ is $1$ and otherwise it is $0$. This corresponds exactly to when there is an edge-hyperwalk of length $1$ from $e_i$ to $e_j$ using $v_q$. The summation over the edges thus correctly counts the number of edge-hyperwalks of length 1 from $e_i$ to $e_j$. 

Suppose the claim holds for $k \ge 1$.

We can view $(\Delta^-)^{k+1}$ as $(\Delta^-)^k \Delta^-$. Thus,
\[ (\Delta^-)^{k+1}(i,j) = \sum_{i=1}^m (\Delta^-)^k (i,q) \cdot \Delta^-(q,j). \]
Any edge-hyperwalk of length $k+1$ may be viewed as an edge-hyperwalk of length $k$ from $e_i$ to some vertex $e_q$ followed by an edge-hyperwalk of length 1 from $e_q$ to $e_j$. Since the number of such edge-hyperwalks are counted by $(\Delta^-)^{k}(i,q)$ and $\Delta^- (q,j)$ respectively, the product of the terms gives the number of edge-hyperwalks of length $k+1$ going from $v_i$ to $v_j$ that have penultimate vertex $v_q$. Summing over all vertices thus gives the total number of edge-hyperwalks of length $k+1$ from $e_i$ to $e_j$.

\end{proof}

\section{Two types of hyperwalks in CW-hypergraphs}\label{sec_CW}

In this section we introduce CW-hypergraphs and different types of walks on edges  of consecutive dimension, as well as counting formulae for them.
First, we recall the definition of finite CW-complexes, which are a class of topological spaces of combinatorial type, for which the combinatorial description of quantum mechanics for graphs~\cite{Mnev2} can be naturally extended.

\begin{definition}\label{def_CW_complex}
A \emph{finite CW-complex} is a topological space $X$, together with a finite family $\{C_{i}\}_{i\in I}$ of subspaces of $X$, called \emph{cells}
such that the following conditions hold:
\begin{enumerate}
\item $X= \amalg_{i\in I}C_i.$
\item For each $n$-dimensional cell $C$ in the partition\footnote{The dimension of a cell is the dimension as a topological space.}, there is a continuous map $f: B^n \to X$  from the $n$-dimensional closed ball into $X$ such that the image of the interior of $B^n$ under $f$ is homeomorphic to $C$. The map $f$ is usually called the \emph{attaching map} of the cell.
\item The image $f(\partial (B^n))$ is contained in the union of some lower dimensional cells of $X$.
\item A subset of $X$ is closed if and only if it intersects each cell in a closed set.
\end{enumerate}
\end{definition}

Finite graphs, simplicial complexes and polytpes are particular instances of CW-complexes. For the purposes of this paper we will consider hypergraphs constructed from finite CW-complexes.

In a CW-complex $X$, we refer to the $d$-dimensional cells as \emph{$d$-cells.} The \emph{$k$-skeleton} of $X$ is the sub-CW-complex consisting of the cells of dimension at most $k$.

\begin{definition}\label{def_CW_hypergraph}

A finite hypergraph $\mathcal H$ is called a \emph{CW-hypergraph} if there is a CW-complex $X$ for which the vertices of $\H$ are the 0-cells of $X$ and the edges of $\H$ are the 0-skeletons of the $n$-cells with $n \ge 1$. The \emph{orientation} of an edge in a CW-hypergraph is given by the ordering of 0-cells in the attaching maps of $X$.

\end{definition}

Let $e$ be an edge of a CW-hypergraph $\mathcal H$. The \emph{dimension} of $e$, denoted by $\dim(e)$, is the dimension of the corresponding cell in the CW-complex $X$. For convience, we refer to the edges of dimension $d$ as \emph{$d$-edges}. Note that a $d$-edge may have more than $d$ vertices

Finite graphs are particular examples of uniform CW-hypergraphs, in which all the edges are 1-dimensional. Every finite CW-complex $X$ has naturally associated a finite graph to it, specifically the 1-skeleton.

Let $\mathcal{H}$ be a CW-hypergraph with vertices $V(\H)=\{v_{1},v_{2},\ldots,v_{l}\}$, $d$-edges $E^d(\H)=\{e_{1}^d,\ldots,e_{q}^d\}$, and $(d+1)$-edges $E^{d+1}(\H) = \{e_1^{d+1},\ldots,e_p^{d+1}\}$. The \emph{$d$-incidence matrix} $I_{d}$ of $\mathcal{H}$ is an $|E^{d}(\H)| \times |E^{d+1}(\H)|$ matrix given by \[I_{d}(i,j)=\begin{cases}
+1 & \text{if } e_{i}^d \subset e_{j}^{d+1} \text{ and they have the same orientation,}\\  
-1 & \text{if } e_{i}^d \subset e_{j}^{d+1} \text{ and they have opposite orientation,}\\
0 & \text{if } e_{i}^d \not\subset e_{j}^{d+1}. 
\end{cases}\] 

Note that the $d$-incidence matrix in this definition is different from the traditional incidence matrix for hypergraphs.

\begin{definition}\label{def_hypergraph_dim}

Let $\mathcal{H}$ be a CW-hypergraph with $I_{d}: \mathbb{C}^{d+1} \rightarrow \mathbb{C}^{d}$ is the $d$-incidence matrix of $\mathcal{H}$. The \emph{even CW-hypergraph Laplacian} with dimension $d$ is
\begin{equation}
\Delta_{d}^{+}:=I_{d}I_{d}^{t}.
\end{equation}
The \emph{odd CW-hypergraph Laplacian} with dimension $d$ is 
\begin{equation}
\Delta_{d}^{-}:=I_{d}^{t}I_{d}.
\end{equation}

\end{definition}

\begin{definition}\label{def_d_d+1}

Let $\mathcal{H}$ be a CW-hypergraph with $V(\H) =\{v_{1},v_{2},\ldots,v_{l}\}$, $E^{d}(\H)=\{e^{d}_{1},\ldots,e^{d}_{q}\}$ and $E^{d+1}(\H) = \{e^{d+1}_{1},\ldots,e^{d+1}_{p}\}$ $(d>0)$. A \emph{$(d,d+1)$-hyperwalk} is a sequence $e^{d}_{i_0},e^{d+1}_{j_1},e^{d}_{i_1},e^{d+1}_{j_2},\ldots,e^{d}_{i_k}$ of hypergraph $d$-edges and $(d+1)$-edges, where repetition is allowed, and each $d$-edges is contained in the $(d+1)$-edges that immediately precede and immediately follow it. 

\end{definition}

The \emph{length} of a $(d,d+1)$-hyperwalk is the number of $(d+1)$-edges. In addition, we consider a $\pm 1$ orientation of hyperwalks. The \emph{sign of} $e^{d}_{i}$ \emph{in} $e^{d+1}_{j}$ is $+1$ if $e^{d}_{i}$ and $e^{d+1}_{j}$ has the same orientation, and $-1$ otherwise. We use $\sgn(e^{d}_{i} \subset e^{d+1}_{j})$ to denote the sign of $e^{d}_{i}$ in $e^{d+1}_{j}$. Then, the \emph{sign} of the $(d,d+1)$-hyperwalk is given by $\prod_{r=1}^{k}(\sgn(e^{d}_{i_{r-1}} \subset e^{d+1}_{j_{r}}) \times \sgn(e^{d}_{i_{r}} \subset e^{d+1}_{j_{r}}))$.

\begin{definition}\label{def_d+1_d}

Let $\mathcal{H}$ be a CW-hypergraph with $V(\H) =\{v_{1},v_{2},\ldots,v_{l}\}$, $E^{d}(\H)=\{e^{d}_{1},\ldots,e^{d}_{q}\}$ and $E^{d+1}(\H) = \{e^{d+1}_{1},\ldots,e^{d+1}_{p}\}$ $(d>0)$. A \emph{$(d+1,d)$-hyperwalk} is a sequence $e^{d+1}_{i_0},e^{d}_{j_1},e^{d+1}_{i_1},e^{d}_{j_2},\ldots,e^{d+1}_{i_k}$ of hypergraph $(d+1)$-edges and $d$-edges, where repetition is allowed, and each $(d+1)$-edges contains the $d$-edges that immediately precede and immediately follow it.
\end{definition}

The \emph{length} of a $(d+1,d)$-hyperwalk is the number of $d$-edges.

In the definition of sign of $e^{d}_{i}$ in $e^{d+1}_{j}$ given in Definition~\ref{def_d_d+1}, the \emph{sign} of the $(d+1,d)$-hyperwalk is given by $\prod_{r=1}^{k}(\sgn(e^{d}_{j_{r-1}} \subset e^{d+1}_{i_{r}}) \times \sgn(e^{d}_{j_{r}} \subset e^{d+1}_{i_{r}}))$. Example~\ref{example} includes a (2,1)-hyperwalk and a (1,2)-hyperwalk in a CW-hypergraph. 

Theorem~\ref{d_d+1_thm} determines the generating function for the number of edge-hyperwalks, which depend on the orientation of the CW-complex. The analogous construction for oriented finite graphs is given in~\cite{Chengzheng} and was inspired by the combinatorial description for quantum mechanics~\cite{Mnev2,delVeccio}, for which the powers of the odd Laplacian should correspond to the odd-dimensional cells of the CW-complex.

\begin{theorem}\label{d_d+1_thm}

$(\Delta_{d}^{+})^{k}(i,j)=\sum_{h,i\rightarrow j,k}\sgn(h)$ where $h,i\rightarrow j,k$ represents the set of $(d,d+1)$-hyperwalks $h$ from $d$-edge $i$ to $d$-edge $j$ of length $k$.

\end{theorem}

\begin{proof}

We proceed by induction on the exponent $k$. Let $E^{d}(\H) = \{e_1^d,\ldots,e_q^d \}$ and $E^{d+1}(\H) = \{e_1^{d+1},\ldots,e_p^{d+1}\}$. When $k=1$, we have $(\Delta_d^+)^1 = \Delta_d^+ = I_d I_d^t$, where $I_d$ is the $d$-incidence matrix. 

Therefore, $\Delta_{d}^{+}(i,j)=\sum_{l=1}^{d}I_{il}I_{jl}$. Since the entries of $I_d$ are from $\{-1,0,+1\}$, each contribution to the sum is either $-1$, $0$, or $1$. If both $e_i^d$ and $e_j^d$ are contained in $e_l^{d+1}$, then we get a non-zero contribution. This contribution is $+1$ if their orientations agree and is $-1$ otherwise. If one, or both, of $e_i^d$ and $e_j^d$ are not contained in $e_l^{d+1}$, then the contribution is $0$. The summation over the signs of walks of length 1 from $e_i^d$ to $e_j^d$ are thus correctly computed by $\Delta_d^+(i,j)$. 

Suppose the claim holds for $k \ge 1$. We can view $(\Delta_d^+)^{k+1}$ as $(\Delta_d^+)^k \Delta_d^+$. Thus, 
\[ (\Delta_d^+)^{k+1}(i,j) = \sum_{l=1}^q (\Delta_d^+)^k(i,l) \cdot \Delta_d^+(l,j). \]

Any $(d,d+1)$-hyperwalk of length $k+1$ may be viewed as a $(d,d+1)$-hyperwalk of length $k$ from $e_i^d$ to some $d$-edge $e_l^d$ followed by a $(d,d+1)$-hyperwalk from $e_l^d$ to $e_j^d$ of length 1. The sign of the combined hyperwalk is the product of the two signs of the $(d,d+1)$-hyperwalk of length $k$ and the $(d,d+1)$-hyperwalk of length 1. Since the sum of signed $(d,d+1)$-hyperwalks of length $k$ and length 1 are counted by $(\Delta_d^+)^k(i,l)$ and $(\Delta_d^+)(l,j)$ respectively, by the Distributive Law, the sum of the signs of the $(d,d+1)$-hyperwalks of length $k+1$ from $e_i^d$ to $e_j^d$ with penultimate $d$-edge $e_l^d$ is given by their product. Summing over all possible penulimate $d$-edges gives the total sum of the signs of $(d,d+1)$-hyperwalks from $e_i^d$ to $e_j^d$.
\end{proof}

Theorem~\ref{d+1_d_thm} computes the number of $(d+1,d)$-hyperwalks on $\mathcal H$ in terms of the powers of the odd Laplacian associated to the $k$-cells in the CW-decomposition associated to the hypergraph $\mathcal H$.

\begin{theorem}\label{d+1_d_thm}

$(\Delta_{d}^{-})^{k}(i,j)=\sum_{h,i\rightarrow j,k}\sgn(h)$ where $h,i\rightarrow j,k$ represents the set of $(d+1,d)$-hyperwalks $h$ from $(d+1)$-edge $i$ to $(d+1)$-edge $j$ of length $k$.

\end{theorem}

\begin{proof}

We proceed by induction on the exponent $k$. Let $E^{d}(\H) = \{e_1^d,\ldots,e_q^d \}$ and $E^{d+1}(\H) = \{e_1^{d+1},\ldots,e_p^{d+1}\}$. When $k=1$, we have $(\Delta_d^-)^1 = \Delta_d^- = I_d^t I_d$, where $I_d$ is the $d$-incidence matrix. 

Therefore, $\Delta_{d}^{-}(i,j)=\sum_{l=1}^{d}I_{il}I_{jl}$. Since the entries of $I_d$ are from $\{-1,0,+1\}$, each contribution to the sum is either $-1$, $0$, or $1$. If both $e_i^{d+1}$ and $e_j^{d+1}$ contain $e_l^{d}$, then we get a non-zero contribution. This contribution is $+1$ if their orientations agree and is $-1$ otherwise. If one, or both, of $e_i^{d+1}$ and $e_j^{d+1}$ do not contain $e_l^{d}$, then the contribution is $0$. The summation over the signs of walks of length 1 from $e_i^{d+1}$ to $e_j^{d+1}$ are thus correctly computed by $\Delta_d^-(i,j)$. 

Suppose the claim holds for $k \ge 1$. We can view $(\Delta_d^-)^{k+1}$ as $(\Delta_d^-)^k \Delta_d^-$. Thus, 
\[ (\Delta_d^-)^{k+1}(i,j) = \sum_{l=1}^q (\Delta_d^-)^k(i,l) \cdot \Delta_d^-(l,j). \]

Any $(d+1,d)$-hyperwalk of length $k+1$ may be viewed as a $(d+1,d)$-hyperwalk of length $k$ from $e_i^{d+1}$ to some $(d+1)$-edge $e_l^{d+1}$ followed by a $(d+1,d)$-hyperwalk from $e_l^{d+1}$ to $e_j^{d+1}$ of length 1. The sign of the combined $(d+1,d)$-hyperwalk is the product of the two signs of the $(d+1,d)$-hyperwalk of length $k$ and the $(d+1,d)$-hyperwalk of length 1. Since the sum of signed $(d+1,d)$-hyperwalks of length $k$ and length 1 are counted by $(\Delta_d^-)^k(i,l)$ and $(\Delta_d^-)(l,j)$ respectively, by the Distributive Law, the sum of signed $(d+1,d)$-hyperwalks of length $k+1$ from $e_i^{d+1}$ to $e_j^{d+1}$ with penultimate $(d+1)$-edge $e_l^{d+1}$ is given by their product. Summing over all possible penulimate $(d+1)$-edges gives the total sum of signed $(d+1,d)$-hyperwalks from $e_i^{d+1}$ to $e_j^{d+1}$.

\end{proof}

\begin{figure}[h]
\centering
\begin{subfigure}{.4\textwidth}
\begin{tikzpicture}
    \tikzset{
    mid arrow/.style={
    decoration={markings,mark=at position 0.5 with {\arrow[scale = 2]{>}}},
    postaction={decorate},
    shorten >=0.4pt}}
	\path (0,0) coordinate (X1); \fill (X1) circle (3pt);
	\path (4,0) coordinate (X2); \fill (X2) circle (3pt);
	\path (2,3.46) coordinate (X3); \fill (X3) circle (3pt);
	\path (2.4,0.8) coordinate (X4); \fill (X4) circle (3pt);
    \path (1.6,1.6) coordinate (X11);
    \path (0.4,3) coordinate (X12);
    \path (2.2,0.4) coordinate (X21);
    \path (2.8,-0.4) coordinate (X22);
    \path (2.8,1.2) coordinate (X31);
    \path (3.6,1.4) coordinate (X32);
	\node[below left] at (X1) {$v_{1}$}; 
	\node[below right] at (X2) {$v_{2}$}; 
	\node[above] at (X3) {$v_{3}$};
	\node[right] at (X4) {$v_{4}$};
    \node[left] at (X12) {$e_{7}$};
    \node[below] at (X22) {$e_{8}$};
    \node[right] at (X32) {$e_{9}$};
	\draw (X1) -- (X2) node[midway,below]{$e_{1}$};
    \draw (X1) -- (X3) node[midway,left]{$e_{2}$};
    \draw (X2) -- (X3) node[midway,right]{$e_{3}$};
	\draw (X1) -- (X4) node[midway]{$e_{4}$};
    \draw (X2) -- (X4) node[midway]{$e_{5}$};
    \draw (X3) -- (X4) node[midway]{$e_{6}$};
    \draw (X11) -- (X12);
    \draw (X21) -- (X22);
    \draw (X31) -- (X32);
	\filldraw[pattern=north west lines,line join=round, pattern color = gray] (X1) to (X3) -- (X4);
    \filldraw[pattern=dots,line join=round, pattern color = gray] (X2) to (X3) -- (X4);
    \filldraw[pattern=north east lines,line join=round, pattern color = gray] (X1) to (X2) -- (X4);
    \filldraw[pattern=north west lines,line join=round, pattern color = gray] (X1) to (X2) -- (X4);
\end{tikzpicture}
      \caption{Hypergraph}
   \label{Figure 1}
\end{subfigure}
\qquad
\qquad
\begin{subfigure}{.4\textwidth}  
\begin{tikzpicture}
    \tikzset{
    mid arrow/.style={
    decoration={markings,mark=at position 0.5 with {\arrow[scale = 2]{>}}},
    postaction={decorate},
    shorten >=0.4pt}}
	\path (0,0) coordinate (X1); 
	\path (4,0) coordinate (X2);
	\path (2,3.46) coordinate (X3);
	\path (2.4,0.8) coordinate (X4);
    \path (1.6,1.6) coordinate (X11);
    \path (0.4,3) coordinate (X12);
    \path (2.2,0.4) coordinate (X21);
    \path (2.8,-0.4) coordinate (X22);
    \path (2.8,1.2) coordinate (X31);
    \path (3.6,1.4) coordinate (X32);
	\node[below left] at (X1) {$v_{1}$}; 
	\node[below right] at (X2) {$v_{2}$}; 
	\node[above] at (X3) {$v_{3}$};
	\node[right] at (X4) {$v_{4}$};
    \node[left] at (X12) {$e^{2}_{1}$};
    \node[below] at (X22) {$e^{2}_{2}$};
    \node[right] at (X32) {$e^{2}_{3}$};
	\draw (X1) -- (X2) node[midway,below]{$e^{1}_{1}$};
    \draw (X1) -- (X3) node[midway,above left]{$e^{1}_{2}$};
    \draw (X2) -- (X3) node[midway,above right]{$e^{1}_{3}$};
	\draw (X1) -- (X4) node[midway, above right]{$e^{1}_{4}$};
    \draw (X2) -- (X4) node[midway, above]{$e^{1}_{5}$};
    \draw (X3) -- (X4) node[midway, below]{$e^{1}_{6}$};
	\filldraw[pattern=north west lines,line join=round, pattern color = gray] (X1) to (X3) -- (X4);
    \filldraw[pattern=dots,line join=round, pattern color = gray] (X2) to (X3) -- (X4);
    \filldraw[pattern=north east lines,line join=round, pattern color = gray] (X1) to (X2) -- (X4);
    \filldraw[pattern=north west lines,line join=round, pattern color=gray] (X1) to (X2) -- (X4);
    \draw (X11) -- (X12);
    \node[scale = 1.5] at (X11) {$\circlearrowright$};
    \draw (X21) -- (X22);
    \node[scale = 1.5] at (X21) {$\circlearrowleft$};
    \draw (X31) -- (X32);
    \node[scale = 1.5] at (X31) {$\circlearrowleft$};
    \draw[mid arrow] (X1) -- (X2);
    \draw[mid arrow] (X1) -- (X3);
    \draw[mid arrow] (X1) -- (X4);
    \draw[mid arrow] (X2) -- (X3);
    \draw[mid arrow] (X2) -- (X4);
    \draw[mid arrow] (X3) -- (X4);
    \fill (X1) circle (3pt);
    \fill (X2) circle (3pt);
    \fill (X3) circle (3pt); 
    \fill (X4) circle (3pt); 
\end{tikzpicture}
\caption{CW-hypergraph}
\label{Figure 2}
\end{subfigure}
\caption{Examples of a hypergraph and a CW-hypergraph} \label{Wholefigure}
\end{figure}
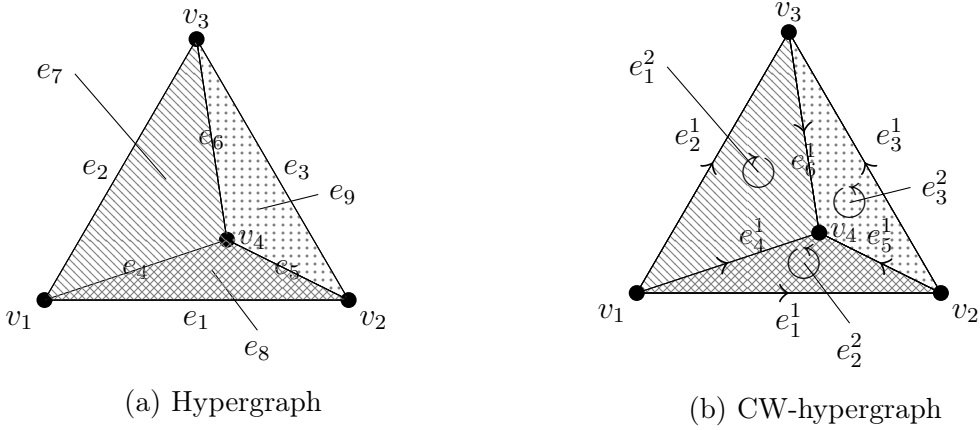

\begin{example}\label{example}
The following are examples of both a hypergraph (Figure~\ref{Figure 1}) and a CW-hypergraph (Figure~\ref{Figure 2}). The hypergraph has nine edges of sizes 2 and 3. The CW-hypergraph is four vertices, six 1-edges and three 2-edges.

For the hypergraph in Figure~\ref{Figure 1}, Theorem~\ref{hyper_walk_thm} shows that there are 5886 hyperwalks from $v_{1}$ to $v_{3}$ of length 4. Similarly, using Theorem~\ref{edge_hyper_walk_thm} we find that there are 384 edge-hyperwalks of length 3 from $e_{7}$ to $e_{9}$. 

For the CW-hypergraph in Figure~\ref{Figure 2}, the sign of $e_6^1$ in $e_1^2$ is negative. An example of a $(1,2)$-hyperwalk from $e^{1}_{1}$ to $e^{1}_{6}$ of length 4 is  $e^{1}_{1},e^{2}_{2},e^{1}_{4},e^{2}_{1},e^{1}_{6},e^{2}_{3},e^{1}_{5},e^{2}_{3},e^{1}_{6}$. This $(2,1)$-hyperwalk has positive sign. On the other hand, $e^{2}_{1},e^{1}_{4},e^{2}_{2},e^{1}_{5},e^{2}_{3}$ is a $(2,1)$-hyperwalk from $e^{2}_{1}$ to $e^{2}_{3}$ of length 2 and with negative sign. In total, Theorem~\ref{d_d+1_thm} shows that the sum of signs of the $(1,2)$-hyperwalks from $e^1_1$ to $e^1_6$ of length 4 is 0. Simiarly, Theorem~\ref{d+1_d_thm} shows that the sum of signs of the $(2,1)$-hyperwalks from $e^2_1$ to $e^2_3$ of length 2 is $+1$.
\end{example}

\section*{Acknowledgements}
We thank the Illinois Geometry Lab (IGL) for the support during the project \emph{Quantum Mechanics for Graphs and CW-Complexes}, from which the results of this paper were derived. I.C thanks Pavel Mn\"ev for fruitful discussions on Graph and CW-quantum mechanics.

\end{document}